\documentclass{article}
\usepackage[letterpaper, portrait, top=1in, left=.8in, right=.8in ]{geometry}
\usepackage{url}
\usepackage{amsmath}
\usepackage{amsfonts}
\usepackage{amsthm}
\usepackage{graphicx}
\usepackage{caption}
\usepackage{subcaption}
\usepackage{color}
\usepackage{enumitem}
\usepackage[ruled]{algorithm2e}

\newcommand{\bc}{\begin{center}}
\newcommand{\ec}{\end{center}}
\newcommand{\bfl}{\begin{flushleft}}
\newcommand{\efl}{\end{flushleft}}

\newcommand{\beqa}{\begin{eqnarray}}
\newcommand{\eeqa}{\end{eqnarray}}
\newcommand{\beqan}{\begin{eqnarray*}}
\newcommand{\eeqan}{\end{eqnarray*}}
\newcommand{\beq}{\begin{equation}}
\newcommand{\eeq}{\end{equation}}
\newcommand{\beit}{\begin{itemize}}
\newcommand{\eeit}{\end{itemize}}

\newcommand{\bit}{\begin{itemize}}
\newcommand{\eit}{\end{itemize}}
\newcommand{\ben}{\begin{enumerate}}
\newcommand{\een}{\end{enumerate}}

\newcommand{\blem}{\begin{lemma}}
\newcommand{\elem}{\end{lemma}}
\newcommand{\bthm}{\begin{theorem}}
\newcommand{\ethm}{\end{theorem}}
\newcommand{\bdefn}{\begin{definition}}
\newcommand{\edefn}{\end{definition}}
\newcommand{\bpf}{\begin{proof}}
\newcommand{\epf}{\end{proof}}
\newcommand{\bcor}{\begin{corollary}}
\newcommand{\ecor}{\end{corollary}}
\newcommand{\bprop}{\begin{proposition}}
\newcommand{\eprop}{\end{proposition}}

\newtheorem{proposition}{Proposition}
\newtheorem{lemma}{Lemma}
\newtheorem{theorem}{Theorem}

\newtheorem{corollary}{Corollary}
\newtheorem{assumption}{Assumption}
\title{\huge{Potential Conditional Mutual Information: Estimators and Properties}\\
		\vspace{8mm}
\author{Arman Rahimzamani and  Sreeram Kannan\thanks{Email: \texttt{armanrz@uw.edu} and \texttt{ksreeram@uw.edu}}
\\Department of Electrical Engineering, \\ University of Washington, Seattle}}

\begin{document}

\maketitle

\begin{abstract}
The conditional mutual information $I(X;Y|Z)$ measures the average information that $X$ and $Y$ contain about each other given $Z$. This is an important primitive in many learning problems including conditional independence testing, graphical model inference, causal strength estimation and time-series problems. In several applications, it is desirable to have a functional purely of the conditional distribution $p_{Y|X,Z}$ rather than of the joint distribution $p_{X,Y,Z}$.  We define the potential conditional mutual information as the conditional mutual information calculated with a modified joint distribution $p_{Y|X,Z} q_{X,Z}$, where $q_{X,Z}$ is a potential distribution, fixed airport. We develop  K nearest neighbor based estimators for this functional, employing importance sampling, and a coupling trick, and prove the finite $k$ consistency of such an estimator. We demonstrate that the estimator has excellent practical performance and show an application in dynamical system inference. 
\end{abstract}

\section{Introduction}
Given three random variables $X,Y,Z$, the conditional mutual information $I(X;Y|Z)$ (CMI) is the expected value of the mutual information between $X$ and $Y$ given $Z$, and can be expressed as follows \cite{cover2012elements},
\beqa CMI_{X \leftrightarrow Y|Z} (p_{X,Y,Z}) := I(X;Y|Z) = D(p_{X,Y,Z} || p_{Z} p_{X|Z} p_{Y|Z}). \eeqa

Thus CMI is a functional of the joint distribution $p_{X,Y,Z}$. A basic property of CMI, and a key application, is the following: $I(X;Y|Z) = 0$ iff $X$ is independent of $Y$ given $Z$. This measure depends on the joint distribution between the three variables $p_{X,Y,Z}$. There are certain circumstances where such a dependence on the entire joint distribution is not favorable, and a measure that depends purely on the conditional distribution $p_{Y|X,Z}$ is more useful. This is because, in a way, conditional independence can be well defined purely in terms of the conditional distribution and the measure $p_{X,Z}$ is extraneous. This motivates the direction that we explore in this paper: we define potential conditional mutual information as a function purely  of $p_{Y|X,Z}$ evaluated with a distribution $q_{X,Z}$ that is fixed a-priori. 

\begin{figure}
\begin{subfigure}[t]{.5\textwidth}
\centering
	\includegraphics[scale=.8]{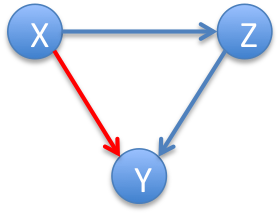}
	\caption{} \label{fig:network}
\end{subfigure}   
\begin{subfigure}[t]{.5\textwidth}
\centering
	\includegraphics[scale=.5]{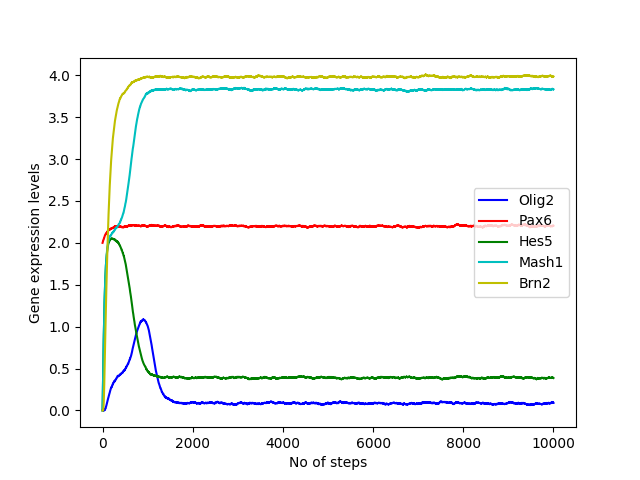}
	\caption{} \label{fig:genes}
\end{subfigure}   
\caption{(a): A causal graph, where the interest is in determining the strength of $X$ to $Y$. (b) A gene expression trace as a function of time for a few example genes.}
\label{fig:network_genes}
\end{figure}

{\em An Example}: Consider the following causal graph where $X \rightarrow Y$, $Z \rightarrow X$ and $Z \rightarrow Y$ shown in Figure~\ref{fig:network}. Let us say $p_{Y|X,Z}$ has a strong dependence on both $X$ and $Z$, say defined by the structural equation $Y = X+Z$. We would like to measure the strength of the edge $X \rightarrow Y$ in this causal graphical model.  One natural measure in this context is $I(X;Y|Z)$. However, if we use $I(X;Y|Z)$ as the strength, the strength goes to zero when $Z \approx X$ and this is undesirable.  In such a case, Janzing et al \cite{janzing2013quantifying} pointed out that a better strength of causal influence is given by the following:
\beqa C(X\rightarrow Y) := D(p_{X,Y,Z} || p_{Z} p_{X} p_{Y|Z}). \eeqa
This causal measure satisfies certain axioms laid out in that paper and is nonzero in the aforesaid example. However, in the case that the distribution $p_{X}$  approaches a deterministic distribution ($X$ is approximately a constant), this measure becomes zero, irrespective of the fact that the relationship from $X$ and $Z$ to $Y$ remains unaltered.  We would like to define a {\em potential dependence measure} that is dependent purely on $p_{Y|X,Z}$ and which has no dependence on the observed $p_{X,Z}$. We note that such a measure should give a (strong) non-zero result if $Y=X+Z$.

{\em The Measure}: We define potential conditional information measure as the conditional mutual information evaluated under a predefined distribution $q_{X,Z}$, and denote it as $qCMI(X \leftrightarrow Y | Z)$, and express it as follows.

\beqa qCMI_{X \leftrightarrow Y | Z}(p_{Y|X,Z}) := CMI_{X \leftrightarrow Y | Z}(q_{X,Z}p_{Y|X,Z}). \eeqa

{\em A Simple Property}: The main question here is how to choose $q_{X,Z}$. A simple property that maybe of interest is the  following, which can be easily stated in case that all three variables $X,Y,Z$ are discrete. In such a case, it will be useful if we can have that  $qCMI_{X \leftrightarrow Y | Z}(p_{Y|X,Z}) = 0$ if and only if $p_{Y|X,Z}$ depends purely only on $Z$. Such a property will be true for qCMI as long as $q_{X,Z}$ is non-zero for every value of $X,Z$. In case that all three variables are real valued, a similar statement can be asserted when $q_{X,Z}$ is a positive everywhere density, under the assumption that $p_{X,Y,Z}$ induces a joint density.

{\em Instantiations}: We will propose some instantiations of the potential CMI by giving examples of the distribution $q_{X,Z}$.
\begin{itemize}
\item CMI: $q_{X,Z} = p_{X,Z} $. Here the $q_{X,Z}$ is the factual measure and hence qCMI devolves to pure CMI. \newline
\item uCMI: $q_{X,Z} = u_{X,Z} $, where $u_{X,Z}$ is the product uniform distribution on $X,Z$. This is well defined when $X,Z$ is either discrete or has a joint density with a bounded support. \newline 
\item nCMI: $q_{X,Z} = n_{X,Z} $, where $n_{X,Z}$ is the i.i.d. Gaussian distribution on $X,Z$. This is well defined when $X,Z$ are real-valued (whether scalar or vector). \newline 
\item maxCMI $ = \max_{q_{XZ}}$ CMI($q_{X,Z} p_{Y|X,Z}$) is defined as an analog of the Shannon capacity in the conditional case, where we maximize the CMI over all possible distributions on $X,Z$. This is akin to tuning the input distribution to maximize the signal in the graph. Note that uCMI or nCMI is not invariant to invertible functional transformations on $X,Z$, whereas maxCMI is indeed invariant to such functional transformations.  \newline 
\item iCMI:  $q_{X,Z} = p_{X}p_Z$ is the CMI evaluated not under the true joint distribution of $X,Z$ but under the product distribution on $X,Z$. This measure is related to the causal strength measure proposed in \cite{janzing2013quantifying}, though not identical. \newline 
\end{itemize}

Note that uCMI, nCMI, maxCMI all satisfy the property that they are zero if and only if $p_{Y|X,Z}$ has no dependence on $X$, whereas CMI and iCMI do not. 

{\em Applications}: A key application of the potential information measures is in testing graphical models, where conditional independence tests are the basic primitive by which models are built \cite{spirtes2000causation,holland1985statistics,dawid1979conditional}.  To give a concrete example of the setting, which motivated us to pursue this line of study, consider the following problem, which can model gene regulatory network inference from time-series data. We observe a set of $n$ time series, $X_i(t)$ for $t=1,2,...T$ with $i=1,2,..n$ and wish to infer the graph of the dynamical system. The underlying model assumption is that $\vec{X}$ is a markov chain with $X_{i}(t)$ depending only on $X_{j}(t-1)$ for $j \in Pa(i)$ and the goal is to determine $Pa(i)$, the set of parents of a given node. This was originally studied in the setting when the variables are jointly Gaussian and hence the dependence is linear (see \cite{granger1969investigating} for the original treatment, and \cite{geiger2015causal,hosseini2016learning} for versions with latent variables). This problem was generalized to the setting with arbitrary probability distributions and temporal dependences in \cite{eichler2012graphical} and studied further in \cite{quinn2015directed}, for one-step markov chains in \cite{sun2015causal} and deterministic relationships in \cite{rahimzamani2016network}. From these works, under some technical condition, we can assert that the following method is guaranteed to be consistent, 
   \begin{equation} x_i \rightarrow x_j \iff  I \{X_i(t-1); X_j(t) | X_{i^c}(t-1) \} > 0.  \end{equation} 
 
 Thus to solve this problem, we estimate the CMI between the aforesaid variables. However, we observed while experimenting with gene regulatory network data (from \cite{qiu2012understanding}), that there is a strange phenomenon; the performance of the inference worsens as we collect more data: the number of data points {\em increases}.

An example of a gene expression time series for a few genes is shown in Figure~\ref{fig:genes}. It is clear that as the number of time points increases, the system is moving into an equilibrium with very little change in gene expression values. This induces a distribution on any $X_i(t)$ which looks more and more like a deterministic distribution.
   

In such a case, an information measure such as CMI which depends on the ``input'' distribution $p_{x_i(t-1)}$ will converge to zero and thus its performance will deteriorate as the number of samples increases. However a measure that depends on the conditional distribution $p_{x_j{t}|x_i(t-1),x_{i}}$ need not deteriorate with increasing number of samples. Thus qCMI is more appropriate in this context (see Sec.~\ref{sec:nonlinear} for performance of qCMI on this problem).

{\em Related work}:
In the case that there is a pair of random variables $X,Y$, recent work \cite{gao2016causal} explored conditional dependence measures which depend only on $p_{Y|X}$. Again in the two-variable case, a measure that had weak dependence on $p_X$ was studied  in \cite{kim2017discovering}. The proposal there was to use the strong data processing constant and hypercontractivity \cite{anantharam2013maximal,polyanskiy2017strong} to infer causal strength; this has strong relationships to information bottleneck \cite{tishby2000information}. In this paper, we extend \cite{gao2016conditional} to conditional independence (rather than independence studied there). In a related but different direction, Shannon capacity, which is a potential dependence metric, was proposed in \cite{kidambi2015shannon} to infer causality from observational data \cite{mooij2016distinguishing}.

{\em Main Contributions}:
In this paper, we make the following main contributions: 
\begin{enumerate}
\item We propose {\em potential conditional mutual information} as a way of quantifying conditional independence, but depending only on the conditional distribution $p_{Y|X,Z}$.
\item We propose {\em new estimators} in the real-valued case that combine ideas from importance sampling, a coupling trick and $k$-nearest neighbors estimation to estimate potential CMI. 
\item We prove that the proposed estimator is {\em consistent} for a fixed $k$, which does not depend on the number of samples $N$. 
\item We demonstrate by simulation studies that the proposed estimator has excellent performance when there are a finite number of samples, as well as an application in gene network inference, where we show that qCMI can solve the non-monotonicity problem. 
\end{enumerate}

\section{Estimator}
In most real settings, we do not have access to either the joint distribution $p_{X,Y,Z}$ or the conditional distribution $p_{Y|X,Z}$, but need to estimate the requisite information functionals from observed samples. We are given $N$ independent identically distributied samples $\{(x_i,y_i,z_i)_{i=1,2,..,N}\}$ from $p_{X,Y,Z}$. In the case of qCMI, the estimator is also given as input the modified distribution $q_{X,Z}$. The estimator needs to estimate qCMI from samples. 

In the case of discrete valued distributions, it is possible to empirically estimate $p_{X,Y,Z}$ from samples and calculate the qCMI from this distribution. We focus our attention here on the case of continuous valued alphabet, where each variable takes on values in a bounded subset of $\mathbb{R}^{d}$. We assume that $X,Y,Z$ are of dimensions $d_x,d_y,d_z$ respectively, and let $f_{X,Y,Z}$ denote the joint density of the three variables (we assume that it exists). In such a case, it is possible to estimate $f_{X,Y,Z}$ using kernel density estimators \cite{devroye1984consistency,sheather1991reliable} and then warp the estimate using the potential measure $q_{X,Z}$. However, it is known that $k$-nearest neighbors based estimators perform better even in the simpler case of mutual information estimation and are widely used in practice \cite{Kra04,khan2007relative}. Therefore in this work, we develop KNN based estimators for qCMI estimation. 

\subsection{Entropy estimation}
Consider first the estimation of the differential entropy of a random variable $X$ with density $f_X$ and observed samples $x_1,...,x_N$.. A simple method to estimate the differential entropy is to use the re-substitution estimator, where we calculate $\hat{h}(X) := \frac{1}{N} \sum_{i=1}^N \log (\hat{f}_X(x_i))$, where $\hat{f}_X$ is an estimate of the density of $X$. We can estimate the density using a KNN based estimator. To do so, we fix $k$ a-priori, and for each sample $x_i$, find the distance $\rho_{k,i}$ to the nearest neighbor. 
\beq \hat{f}_X(x_i) c_d \rho_{k,i}^d \approx \frac{k}{N}. \eeq
This estimator is not consistent when $k$ is fixed, and it was shown in a remarkable result by Kozhachenko and Loenenko \cite{kozachenko1987sample} that the bias is independent of the distribution and can be computed apriori. Thus the following estimator was shown in \cite{kozachenko1987sample} to be consistent for differential entropy. 
\begin{eqnarray*} \hat{h}_{\rm KL}(X)& = & \frac{1}{N} \sum_{i=1}^N  \log \frac{N\rho_{k,i}^d c_d}{k} +  \log k -\psi(k). \end{eqnarray*}

While it is possible to have estimators which fix an $\epsilon$ apriori and then find the number of nearest neighbors to plug into the formula, such estimators do not adapt to the density (some regions will have many more points inside an $\epsilon$ neighborhood  than others) and do not have a consistency proof as well. We mention this as fixed $\epsilon$ estimators are used for a sub-problem in our estimator.

\subsection{Coupling trick}
The conditional mutual information can be written as a sum of $4$ differential entropies, and one can estimate these differential entropies independently using KNN estimators and sum them.
\begin{eqnarray*} I(X;Y|Z) = -h(X,Y,Z) - h(Z) + h(X,Z) + h(Y,Z).  \end{eqnarray*} 
However, even in the case of mutual information, the estimation can be improved by an inspired coupling trick, in what is called the KSG estimator \cite{Kra04}. We note that the original KSG estimator did not have a proof of consistency and its consistency and convergence rates were analyzed in a recent paper \cite{gao2017demystifying}. Also of interest is the fact that the coupling trick has been shown to be quite useful in problems where $X$, $Y$ or both have a mixture of discrete and continuous distributions or components \cite{gao2017estimating}.  

This trick was applied in the context of conditional mutual information estimation in  \cite{frenzel2007partial}. However, we note that this estimator of CMI does not have a proof of consistency to the best of our knowledge. The CMI estimator essentially fixes a $k$ for the $(X,Y,Z)$ vector and calculates for each sample $(x_i,y_i,z_i)$, the distance $\rho_{k,i}$ to the $k$-th nearest neighbor. The estimator fixes this $\rho_{k,i}$ as the distance and calculates the number of nearest neighbors within $\rho_{k,i}$ in the $Z$, $(X,Z)$ and $(Y,Z)$ dimensions as $n_{z,i}, n_{xz,i}, n_{yz,i}$ respectively. The CMI estimator is then given by,
\beqa \hat{CMI} :=  \frac{1}{N} \sum_{i=1}^N \left( \psi(k) - \log(n_{xz,i}) - \log(n_{yz,i}) + \log(n_{z,i}) \right)  + \log \left(\frac{c_{d_x+d_z} c_{d_y+d_z}}{c_{d_x+d_y+d_z}{c_{d_z}}} \right). \eeqa

\subsection{qCMI estimator}
Here, we adapt this estimator to calculate the qCMI for a given potential distribution $q_{X,Z}$. The major difference is the utilization of an importance sampling estimator to get the importance of each sample $i$ estimated as follows,
\beqa \omega_i :=\frac{q_{XZ}(x_i,z_i)}{\hat{f}_{XZ}(x_i,z_i)}. \eeqa

However, importance sampling based reweighting alone is insufficient to handle qCMI estimation, since there is a logarithm term which depends on the density also. We handle this effect by appropriately re-weighting the number of nearest neighbors for the $(y,z)$ and $z$ terms carefully using the importance sampling estimators. The estimation algorithm is described in detail in Algorithm~\ref{algo:qCMI}.
\vspace{.2in}

\begin{algorithm}[H]
\vspace{.1in}
\KwData{Data Samples $(x_i,y_i,z_i)$ for $i=1, \ldots, N$ and $q_{X,Z}$}
\KwResult{$\hat{qCMI}$ an estimate of $qCMI$}
\vspace{.1in}
\textbf{Step 1:} Calculate weights $\omega_i$ \\
\For{$i=1, \ldots, N$}{
  Estimate $\hat{f}_{XZ}(x_i,z_i)$ using a Kernel density estimator \cite{sheather1991reliable,devroye1984consistency}.\\
  $\omega_i :=\frac{q_{XZ}(x_i,z_i)}{\hat{f}_{XZ}(x_i,z_i)}$, the importance sampling estimate of sample $i$.}
\vspace{.1in}
\textbf{Step 2:} Calculate information samples $I_i$ \\
\For{$i=1, \ldots, N$}{
$\rho_{k,i}$ := Distance of $k$-th nearest neighbor of $(x_i,y_i,z_i)$. \\
$n_{xz,i} := \sum_{j\neq i: \| (x_i,z_i)-(x_j,z_j) \| < \rho_{k,i}} 1$, the number of neighbors of $(x_i,z_i)$ within distance $\rho_{k,i}$. \\
$n_{yz,i} :=  \sum_{j\neq i: \| (y_i,z_i)-(y_j,z_j) \| < \rho_{k,i} }^N \omega_j  $, the {\em weighted} number of  neighbors of $(y_i,z_i)$ within distance $\rho_{k,i}$.  \\
$n_{z,i} := \sum_{j\neq i: \| z_i-z_j \| < \rho_{k,i}} \omega_i$, the {\em weighted} number of neighbors of $z_i$ within distance $\rho_{k,i}$. \\
$I_i :=  \psi(k) - \log(n_{xz,i}) - \log(n_{yz,i}) + \log(n_{z,i})$. 
}
\vspace{.1in}
\textbf{Return} $\hat{qCMI} = \frac{1}{N}\sum_{i=1}^N \omega_i I_i + \log \left(\frac{c_{d_x+d_z} c_{d_y+d_z}}{c_{d_x+d_y+d_z}{c_{d_z}}} \right)$.
\vspace{.1in}

\caption{qCMI algorithm} \label{algo:qCMI}
\end{algorithm}

\section{Properties}
Our main technical result is the consistency of the proposed potential conditional mutual information estimator. This proof requires combining several elements from importance sampling, and accounting for the correlation induced by the coupling trick, in addition to handling the fact that the $k$ is fixed and hence introduces a bias into estimation.  

\begin{assumption} \label{assumptions}
We make the following assumptions.

\begin{enumerate}[label=\alph*)]
\item $\int f_{XYZ}(x,y,z) \left( \log f_{XYZ}(x,y,z) \right)^2 dx dy dz < \infty$. 
\item All the probability density functions (PDF) are absolutely integrable, i.e. for all $A,B \subset \{X,Y,Z\} $, $\int | f_{A|B}(a|b) |  da < \infty$ and $\int | f_{A|B}^q(a|b) |  da < \infty$.
\item There exists a finite constant $C$ such that the hessian matrices of $f_{XYZ}$ and $f^q_{XYZ}$ exist and it's true that $\max \{  \|h(f_{XYZ})\|_2, \| h(f^q_{XYZ}) \|_2 \} < C$ almost everywhere. 
\item  All the PDFs are upper-bounded, i.e. there exists a positive constant $C'$ such that for all $A,B \subset \{X,Y,Z\}$, $f_{A|B}<C'$ and $f^q_{A|B}<C'$ almost everywhere.
\item $f_{XZ}$ is upper and lower-bounded, i.e. there exist positive constants $C1$ and $C2$ such that $C_1 q_{XZ}(x,z)<f_{XZ}(x,z)< C_2 q_{XZ}(x,z)$ almost everywhere.
\item There bandwidth $h_N$ of kernel density estimator is chosen as $h_N=\frac{1}{2}N^{-1/(2d_x+2d_z+3)}$.
\item The k for the KNN estimator is chosen satisfying $k > \max \{\frac{d_z}{d_x+d_y}, \frac{d_x+d_y}{d_z},\frac{d_x+d_z}{d_y} \}$
\end{enumerate}
\end{assumption}

\begin{theorem}\label{theorem:ucmi_convergence}
Under the Assumption \ref{assumptions}, the qCMI estimator expressed in Algorithm~\ref{algo:qCMI} converges to the true value qCMI.
\begin{equation} \hat{qCMI} \overset{p}{\to}  qCMI \end{equation}
\end{theorem}
\bpf 
Please see Section~\ref{sec:proof} for the proof.
\epf

\section{Simulation study}
In this section, we describe some simulated experiments we did to test the qCMI algorithm. The reader should notice that all the tests we have done are taking $q_{XZ}$ as $u_{XZ}$, i.e. all the tests are done for the special case of uCMI. So by exploiting the qCMI notations we mean uCMI everywhere.

\subsection{qCMI consistency}

The first numerical experiment we do is to test the consistency of our qCMI estimation algorithm. We set up a system of three variables $X$, $Y$ and $Z$. The variables $X$ and $Z$ are independent taken from $u^{n}(0,1)$ distribution, i.e. $X$ and $Z$ are taken form a uniform distribution and then raised to a power of $n$. When $n=1$, the variables $X$ and $Z$ are already uniform. When $n$ is large, the $u^{n}(0,1)$ distribution skews more towards $0$. For simplicity we apply identical $n$ for both $X$ and $Z$ here. Then $Y$ is generated as $Y = (X + Z + W) \mod 1$ in which the noise term $W$ is sampled from $u(0,0.2)$. From elementary information theory calculation, we can deduce that $I(X;Y|Z)=\log\left(  \frac{1}{.2} \right)=1.609$ if $n=1$. Thus $I^q(X;Y|Z)=1.609$ for all $n$. We plot the estimated value against the ground truth.

\begin{figure}
\begin{subfigure}[t]{.33\textwidth}
\centering
	\includegraphics[scale=.4]{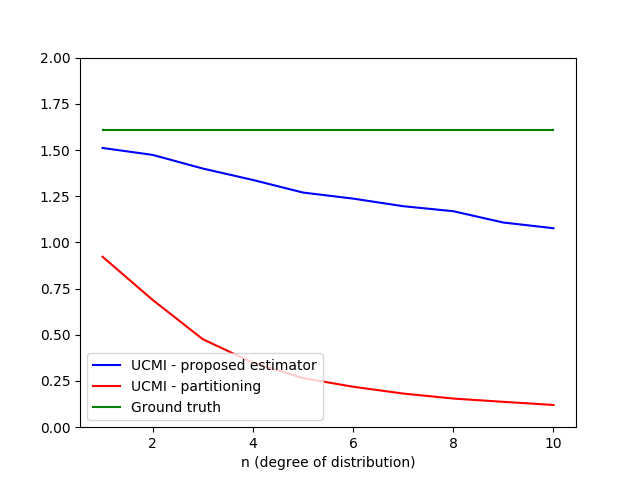}
	\caption{}\label{fig:ucmi_vs_degree1}
\end{subfigure}   
\begin{subfigure}[t]{.33\textwidth}
\centering
	\includegraphics[scale=.4]{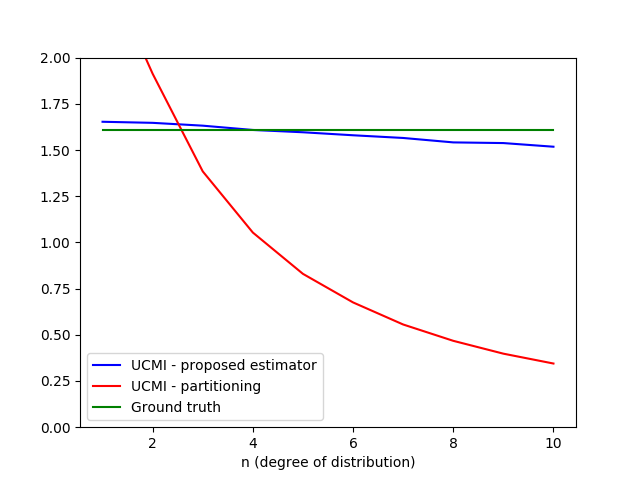}
	\caption{}\label{fig:ucmi_vs_degree2}
\end{subfigure}   
\begin{subfigure}[t]{.33\textwidth}
\centering
\includegraphics[scale=.4]{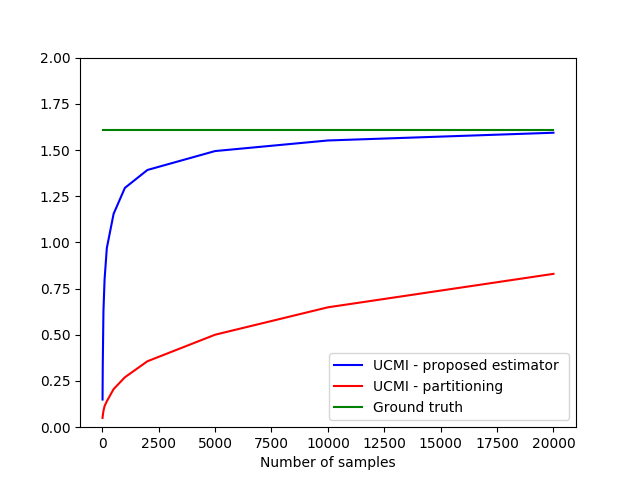}
	\caption{}\label{fig:ucmi_vs_samples}
\end{subfigure}   

\caption{The qCMI values calculated for $u^{n}(0,1)$ distributions for $X$ and $Z$, and uniform noise: (a) N=1000 samples and (b) N =20000 samples. (c) Degree $n=5$.}
\label{fig:ucmi_vs_degree}
\end{figure}

As the first part of the experiment, we keep the number of samples constant at $1000$ and $20000$, and change the degree $n$ from $1$ to $10$. We compare the results of our KSG-based method with the simple partitioning method, and the theoretical value of qCMI. For the partitioning method, the number of partitions at each dimension is determined by $\sqrt[3]{100N}$, so that we observe on average $100$ samples  inside each quantization bin.

The results are shown in Figure~\ref{fig:ucmi_vs_degree1} and Figure~\ref{fig:ucmi_vs_degree2}. Our expectation is that qCMI remains constant as $n$ (degree of distribution) changes. We see that with relatively high number of samples, the accuracy of proposed qCMI is satisfactorily high.   

As the second part of the experiment, we do the same experiment as the first part, but this time we keep $n=5$ and change the number of samples. The result is shown in Figure  \ref{fig:ucmi_vs_samples}. We can see convergence of KSG-based qCMI estimator to the true value and how it outperforms the partitioning-based qCMI method.  


As the third part of the experiment, we repeated the process for the first part, but replaced the $u^{n}(0,1)$ distributions with $\beta(1.5,1.5)$ and the noise distribution with $N(0,\sigma^2)$ and repeated the experiment for $\sigma=0.3,1.0$. For this part we kept the number of partitions at $25$ for each dimension. The results of calculated qCMI values are shown in Figure~\ref{fig:beta0.3} and Figure~\ref{fig:beta1.0}. 

\subsection{Dealing with discrete components}

As we discussed before, the qCMI algorithm replaces the observed distribution $f_{XZ}$ distribution with a distribution $q_{XZ}$. This property comes in handy when we want to remove the bias caused by repeated samples. For example, as discussed earlier, suppose that we want to measure the mutual information of two coupled variables in a dynamical system evolving through time. Such systems usually start from an initial state, go through a transient state and eventually reach a steady state. If one takes samples of the system's state at a constant rate to study the interaction of two variables, they might end up taking too many samples from the initial and steady states while the transient phase which usually happens in a relatively short time might be more informative. The conditional mutual information is not able to deal with this undesirable bias caused by the initial and steady states, while qCMI inherently deals with the effect by compensating for the samples which are less likely to happen.

To better observe the effect, we repeat the first experiment of the previous section, but this time we generate $1000$ samples from the scenario, and then add zeros to the $X$, $Y$ and $Z$ to create a high probability of occurrence at $(0,0,0)$. The proof of consistency of the estimator holds only when there is a joint density, i.e., the joint measure is absolutely continuous with respect to the Lebesgue measure, and hence does not directly apply to this case. We refer the reader to \cite{gao2017estimating} for an analysis of a similar coupled KNN estimator for mutual information in the discrete-continuous mixture case. 

Changing the number of zero points added from $0$ to $20000$, we apply the conditional MI and qCMI to the data generated and compare the results. As we can see in figure \ref{fig:ucmi_vs_zeros}, with the number of zeros increasing, the value of conditional MI falls down to zero, unable to capture the inter-dependence of $X$ and $Y$ given $Z$, while qCMI value remains unchanged, properly discovering the inter-dependence from the transient values.

%


\begin{figure}
\begin{subfigure}[t]{.33\textwidth}
\centering
	\includegraphics[scale=.4]{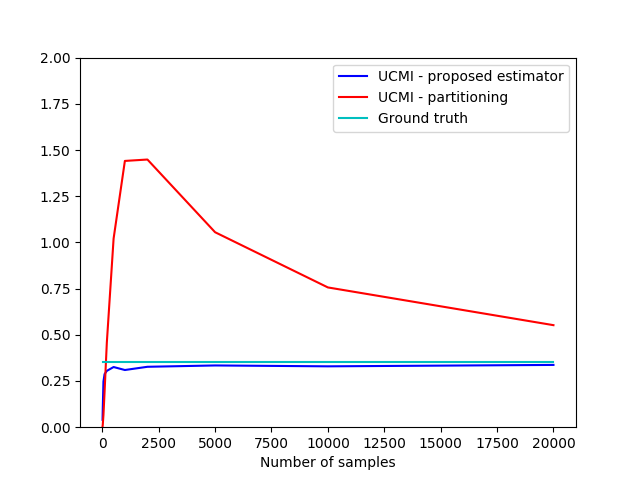}
	\caption{} \label{fig:beta0.3}
\end{subfigure}   
\begin{subfigure}[t]{.33\textwidth}
\centering
	\includegraphics[scale=.4]{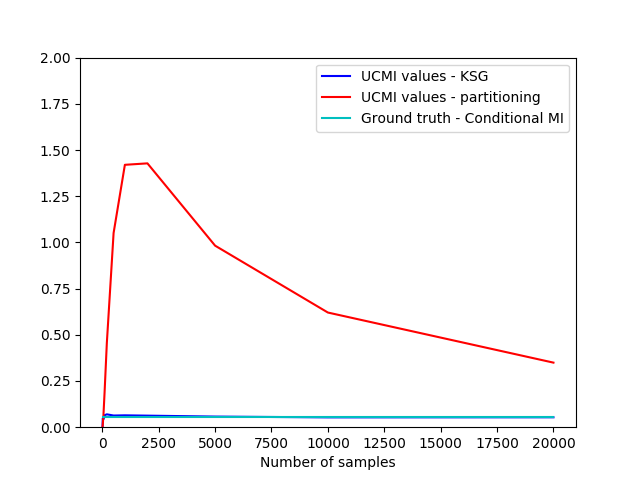}
	\caption{} \label{fig:beta1.0}
\end{subfigure}
\begin{subfigure}[t]{.33\textwidth}
\centering
\includegraphics[scale=.4]{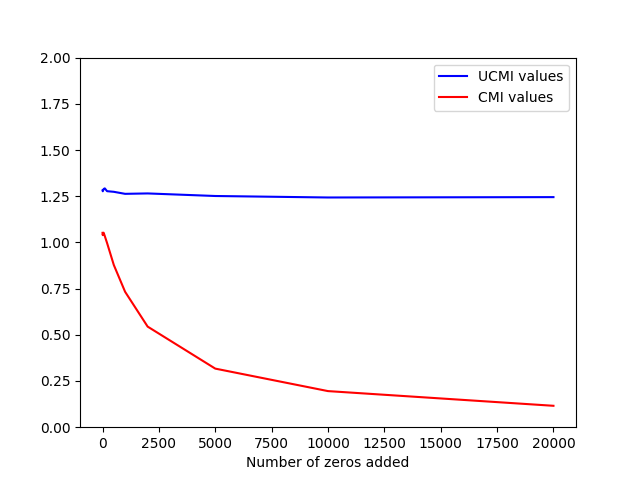}
	\caption{}\label{fig:ucmi_vs_zeros}
\end{subfigure}

\caption{The qCMI values calculated for a system with beta distribution for $X$ and $Z$ and Gaussian additive noise: (a) $\sigma=0.3$, and (b) $\sigma=1.0$. (c) The qCMI and CMI values versus the number of zeros added.}
\label{fig:ucmi_vs_samples_betagaussian}
\end{figure}

\subsection{Non-linear Neuron Cells' Development Process}  \label{sec:nonlinear}

In this section, we apply the RDI and uRDI algorithms to neuron cells' development process simulated based on a model from \cite{qiu2012understanding} which can be modeled as a dynamical system. A dynamical System is described as a set of variables shown by a vector of $\underline{x}$ which evolve through time starting from an initial state $\underline{x}(0)$. The evolution can be described as a vector function $\underline{g}(.)$ such that $\underline{x}(t)=\underline{g}\left( \underline{x}(t-1) \right)$. Note that $\underline{g}$ can be a stochastic function in general, i.e. it may include random coefficients, additive noise and so on.

The dynamical system here describes the evolution of 13 genes through the development process. The non-linear equations governing the development process approximate a continuous development process, in which $\dot{\underline{x}}(t)=g(\underline{x}(t-1))$. In other words, $\underline{x}(t)=\underline{x}(t-1)+dt.g \left( \underline{x} (t-1) \right) + \underline{n}(t)$ in which $ \underline{n}$ are independent Gaussian noises $\sim N(0,\sigma^2)$. 

For this system, we want to infer the true network of causal inferences. In a dynamical system, we say $x_i$ \textit{causes} $x_j$ if $x_j(t)$ is a function of $x_i(t-1)$. For this purpose, we first apply the RDI algorithm \cite{rahimzamani2016network} to extract the pairwise directed causality between the variables by calculating $I\left( x_i(t-1), x_j(t) | x_j(t-1) \right)$. Then we apply the \text{uRDI} algorithm, in which the conditional mutual information $I(X;Y|Z)$ in RDI is replaced with qCMI as $I^q(X;Y|Z)$ using $q_{X,Z}$ as a uniform distribution.

This system is a good example of a system in which the genes undergo a rather short transient state compared to the initial and steady states, and hence we expect an improvement in the performance of causal inference by applying uRDI (see Figure~\ref{fig:genes} for an example run of the system). The details of the dynamical system are given in \cite{qiu2012understanding}.

We simulated the system for discretization $dt=0.1$ and $\sigma = .001$, and changed the number of steps until which the system continues developing, and then applied the RDI and uRDI algorithms to evaluate the performance of each of the algorithms in terms of the area-under-the-ROC-Curve (AUC). The results are shown in Figure \ref{fig:auc_vs_steps_nonlinear}. As we can see, with the number of steps increasing implying the number of samples captured in the steady state are increased, the uRDI algorithm outperforms RDI. In another test scenario, we fixed the number of steps at $200$, but concatenated several runs of the same process. The results and the improvement of performance by uRDI can be seen in the Figure \ref{fig:auc_vs_numberofruns_nonlinear}.

\subsection{Decaying Linear Dynamical System}

In this section, we simulate a linear decaying dynamical system.  A dynamical system in the simple case of a deterministic linear system can be described as: 
\begin{equation}\label{eq:linear_system} \underline{x}(t)= A \underline{x}(t-1) \end{equation} 
In which $A$ is a square matrix.

Here we simulate a system of $13$ variables, all of them initialized from a $u(0.5,2)$ distribution. The first $6$ variables $(x_1, \ldots, x_6)$ are evolved through a linear deterministic process as in \eqref{eq:linear_system} in which $A$ is a square $6 \times 6$ matrix initialized as:
\begin{equation}
A =
\begin{bmatrix}
u(0.75,1.25) & 0 & 0 & 0 & u(0.75,1.25) & 0 \\
u(0.75,1.25) & u(0.75,1.25) & 0 & 0 & 0 & u(0.75,1.25) \\
0 & u(.75,1.25) & u(.75,1.25) & 0 & 0 & 0 \\
0 & u(.75,1.25) & 0 &  u(.75,1.25) & 0 & 0 \\
0 & 0 & u(.75,1.25) & u(.75,1.25) & u(.75,1.25) & 0 \\
u(.75,1.25) & u(.75,1.25) & 0 & 0 & 0 & u(.75,1.25)
\end{bmatrix}
\end{equation}
Then the matrix $A$ is divided by $5*\lambda_{\max}(A)$ in which $\lambda_{\max}(A)$ is the greatest eigenvalue of $A$. It's done to make sure that all the variables decay exponentially to $0$. After initialization, the matrix $A$ is kept constant throughout the development process, i.e. it doesn't change with time $t$.

The other 7 variables $(x_7, \ldots, x_{13})$ are random independent Gaussian variables.

In this experiment, we simulate the system described above for various numbers of time-steps, keeping the standard deviation of the Gaussian variables at $\sigma=0.1$, and applied both RDI and uRDI algorithms to infer the true causal inferences. Then we calculate the AUC values, the results are shown in Figure~\ref{fig:auc_vs_steps_linear}. As we can see, the uRDI algorithm outperforms RDI by a margin of $0.1$ in terms of AUC.

\begin{figure}
\begin{subfigure}[t]{.33\textwidth}
\centering
	\includegraphics[scale=.4]{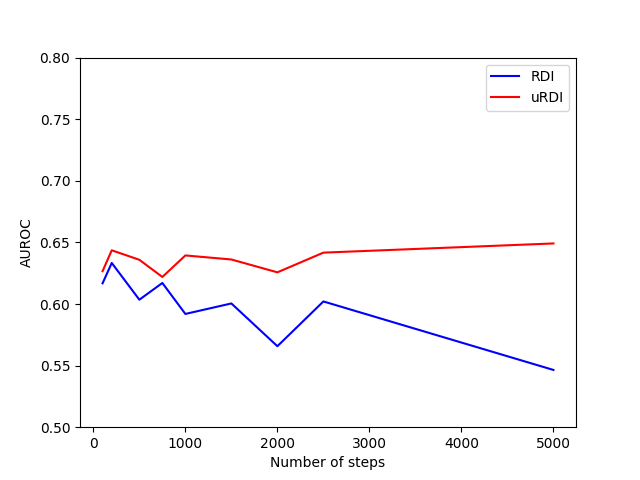}
	\caption{}
	\label{fig:auc_vs_steps_nonlinear}
\end{subfigure}   
\begin{subfigure}[t]{.33\textwidth}
\centering
	\includegraphics[scale=.4]{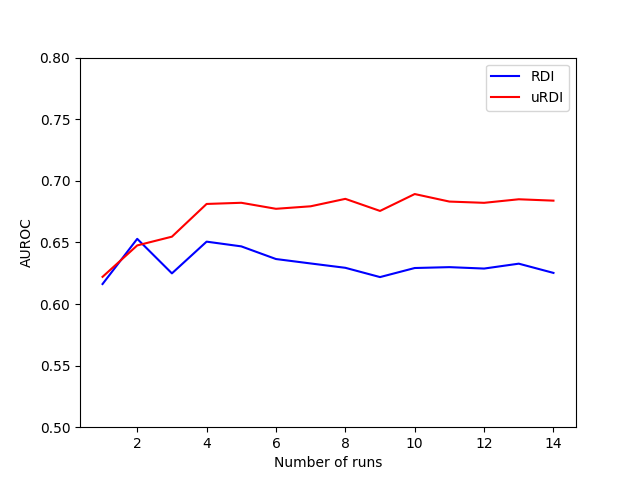}
	\caption{}
	\label{fig:auc_vs_numberofruns_nonlinear}
\end{subfigure}
\begin{subfigure}[t]{.33\textwidth}
\centering
\includegraphics[scale=.4]{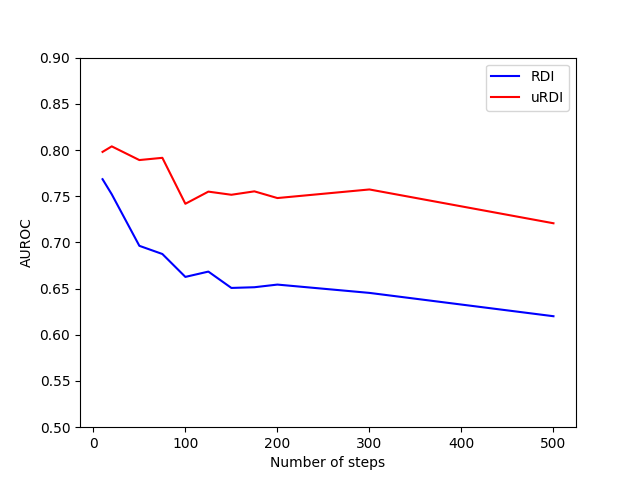}
\caption{}
\label{fig:auc_vs_steps_linear}
\end{subfigure} 
\caption{AUC values  for the neuron cells' development process: a) versus the number of steps. b) versus the number of runs. (c) for the decaying linear system}
\end{figure}

\section{ Future Directions}
In this section, we will describe some promising directions for further investigation.

\begin{enumerate}
\item {\em Quantifying causal strength}: As pointed out earlier, potential conditional mutual information can be used as a metric for quantifying causal strength when the graph is a simple three node network (shown in Figure~\ref{fig:network}). However, further work is needed in order to generalize the definition to deduce the causal strength of an edge or a set of edges in an arbitrary graph, akin to the formulation in \cite{janzing2013quantifying} and to study the relative advantages and disadvantages of such a formulation. 

\item {\em Discrete qCMI estimators}:  It has been shown in recent work that such estimators are not optimal even for determining mutual information in the discrete alphabet case \cite{valiant2011estimating, jiao2015minimax, wu2016minimax}. A very interesting question is how such minimax-rate optimal estimators can  be developed in the potential measures problem. 

\item {\em maxCMI estimation}: While we have developed efficient estimators for qCMI, in maxCMI, there is a further maximization over potential distributions $q$, which leads to some interesting interactions between estimation and optimization. Recent work has studied estimation of Shannon capacity on continuous alphabets, however, the formulation is not convex leading to possible local minima \cite{gao2016causal}. Further work is needed in order to find provably optimal estimators for maxCMI in the continuous case.  

\item {\em Other conditional measures}: Recent work \cite{kim2017discovering} has used strong data processing constants as a way for quantifying dependence between two variables, with relationships to information bottleneck. These measures depend {\em partially} on the factual measure $p_X$, and are implicitly regularized. One direction of future work is to develop multi-variable versions of such estimators to estimate the strength of conditional independence, for example. 

\item {\em Multivariable measures}: Develop estimators that can handle more general multi-variable information measures including total correlation \cite{watanabe1960information} and multi-variate mutual information \cite{chan2015multivariate}. 

\item {\em Ensemble estimation}: Another approach exploiting k-nearest-neighbors for mutual information is the so-called ensemble estimation approach, where estimators for different $k$ are combined together to get a stronger estimator, with fast convergence \cite{moon2017ensemble}. An interesting direction of research is to obtain ensemble estimators for potential measures. 

\end{enumerate}

\section{Acknowledgement}
The authors would like to thank Xiaojie Qiu and Cole Trapnell for discussions that triggered this work. This work was supported in part by NSF Career award (grant 1651236) and NIH award number R01HG008164.

\appendix
\section{Proof} \label{sec:proof}
Here we'll try to introduce a proof for qCMI algorithm we devised. As we know, the conditional mutual information is defined as,
\begin{equation} I(X;Y|Z) = \int f_{XYZ}(x,y,z) \log \left( \frac{f_{Y|XZ}(y|x,z)}{f_{Y|Z}(y|z)} \right) dx dy dz.\end{equation}

The qCMI is defined as the mutual information of $X$ and $Y$ given $Z$ when the joint distribution of $X$ and $Z$ is replaced by a joint uniform distribution $q_{XZ}(x,z)$,
\begin{equation} I^q(X;Y|Z) = \int f_{Y|XZ}(y|x,z) q_{XZ}(x,z) \log \left( \frac{f_{Y|XZ}(y|x,z)}{f_{Y|Z}^q(y|z)} \right) dx dy dz. \end{equation}

In which:
\begin{eqnarray} 
f_{Y|Z}^q(y|z) &=& \frac{f_{YZ}^q(y,z)}{f_Z^q(z)}, \\
f_{YZ}^q(y,z) &=& \int f_{XYZ}^q(x,y,z) dx, \\
f_{Z}^q(z) &=& \int f_{XYZ}^q(x,y,z) dx dy.
\end{eqnarray}

from now on, the superscript $U$ over each distribution function implies that the actual $f_{XZ}(x,z)$ is replaced by $q_{XZ}(x,z)$.  Equivalently, qCMI can be written as,
\begin{equation} I^q(X;Y|Z) = -h^q(X,Y,Z) + h^q(X,Z) + h^q(X,Y) - h^q(Z), \end{equation}
where,
\begin{eqnarray}
 h^q(X,Y,Z) & \equiv &  -\int f_{XYZ}^q(x,y,z) \log f_{XYZ}(x,y,z)  dx dy dz. \\
 h^q(X,Z) & \equiv &  -\int f_{XYZ}^q(x,y,z) \log f_{XZ}(x,z) dx dy dz. \\
 h^q(Y,Z) & \equiv &  -\int f_{XYZ}^q(x,y,z) \log f_{YZ}^q(y,z) dx dy dz. \\
 h^q(Z) & \equiv &  -\int f_{XYZ}^q(x,y,z) \log f_Z^q(z) dx dy dz.
\end{eqnarray}
 
Note that $-h^q(X,Y,Z) + h^q(X,Z)=-h^q(Y|X,Z) =  \int f_{XYZ}^q(x,y,z) \log \left( f_{Y|XZ}(y|x,z) \right) dx dy dz $. So the term inside the logarithm is independent of the distribution over $(X,Z)$ and hence $\log f_{XYZ}(x,y,z)$ and $\log f_{XZ}(x,z)$ appear when defining $h^q(X,Y,Z)$  and $ h^q(X,Z)$.

Here we introduce a qCMI estimator, based on the KSG estimator for \textit{UMI}. Remember the KSG-type estimator for the conditional MI,
\begin{equation}
\hat{I}(X;Y|Z) = \frac{1}{N} \sum_{i=1}^N \left( \psi(k)-\log(n_{xz,i})-\log(n_{yz,i})+\log(n_{z,i}) \right) +  C(d_x,d_y,d_z).
\end{equation}
where 
\beqa C(d_x,d_y,d_z) :=  \log \left(\frac{c_{d_x+d_z} c_{d_y+d_z}}{c_{d_x+d_y+d_z}{c_{d_z}}} \right). \eeqa

The estimator for the qCMI is as below,
\begin{equation}\label{eq:qCMI}
\hat{I}^q(X;Y|Z) = \frac{1}{N} \sum_{i=1}^N \omega_i g_i(x_i,y_i,z_i),
\end{equation}
where,
\begin{eqnarray} 
\omega_i & \equiv & \frac{q_{XZ}(x_i,z_i)}{\hat{f}_{XZ}(x_i,z_i)}. \\
g_i(x_i,y_i,z_i) &\equiv & \psi(k)-\log(n_{xz})-\log(\sum_{\| (y_i-y_j,z_i-z_j) \|<\rho_{k,i}} \omega_j)+\log(\sum_{\| z_i-z_j \|<\rho_{k,i}} \omega_j) + C(d_x,d_y,d_z).
\end{eqnarray}

\section {KSG estimator for qCMI: Proof of convergence}

 Similar to the \cite{gao2016conditional} we define,
\begin{eqnarray}
 \omega_i' & \equiv & \frac{q_{XZ}(x_i,z_i)}{f_{XZ}(x_i,z_i)}. \\
 n_{yz,i}' & \equiv & \sum_{j \in N_{yz}^{\epsilon(i)}} \omega_j' .\\
 n_{z,i}' & \equiv & \sum_{j \in N_{z}^{\epsilon(i)}} \omega_j' .\\
g'(x_i,y_i,z_i) & \equiv &  \psi(k)-\log(n_{xz})-\log( n_{yz,i}')+\log( n_{z,i}') + C(d_x,d_y,d_z).
\end{eqnarray}

From the triangle inequality, we can write,
\begin{eqnarray}
 & & | \hat{I}^q(X;Y|Z) -  I^q(X;Y|Z) | \\
& \leq & | \hat{I}^q(X;Y|Z) - \frac{1}{N} \sum_{i=1}^N \omega_i' g'(x_i,y_i,z_i) | \label{eq:ucmi_first_line} \\
& + & | \frac{1}{N} \sum_{i=1}^N \omega_i' g'(x_i,y_i,z_i) - I^q(X;Y|Z) |. \label{eq:ucmi_second_line}
\end{eqnarray}

To show the convergence of the KSG estimator for qCMI, we will show that \eqref{eq:ucmi_first_line} and \eqref{eq:ucmi_second_line} both converge to zero. The Lemma \ref{lemma:first_term_convergence} proves that \eqref{eq:ucmi_first_line} converges to zero. The lemma will be proven through the Section \ref{sec:first_term_convergence}.
\begin{lemma} \label{lemma:first_term_convergence}
The term \eqref{eq:ucmi_first_line} converges to $0$ as $N \longrightarrow \infty$ in probability.  
\end{lemma}

 For \eqref{eq:ucmi_second_line} we will first write the left hand side term as four entropy terms and show the convergence of each of the terms to their respecttive true entropy term. This will be done in the Lemmas \ref{lemma:hxyz_convergence} and \ref{lemma:3_terms_convergence}. So we can write,
\begin{equation}\frac{1}{N} \sum_{i=1}^N \omega_i' g'(x_i,y_i,z_i) = \hat{h}^q(Y,Z) + \hat{h}^q(X,Z) - \hat{h}^q(X,Y,Z) - \hat{h}^q(Z)  
- \sum_{i=1}^N \frac{\omega'_i}{N} \left( \log(N-1) + \psi(N) \right), \end{equation}
where,
\begin{eqnarray}
 \hat{h}_N^q(X,Y,Z) & \equiv & \frac{1}{N} \sum_{i=1}^N \omega_i' \left( -\psi(k) + \psi(N) + \log c_{d_x+d_y+d_z} + (d_x+d_y+d_z) \log \rho_{k,i} \right). \label{eq:h_xyz} \\
 \hat{h}_N^q(X,Z) & \equiv & \frac{1}{N} \sum_{i=1}^N \omega_i' \left( -\log(n_{xz,i}) + \log(N-1) + \log c_{d_x+d_z} + (d_x+d_z) \log \rho_{k,i} \right). \\
 \hat{h}_N^q(Y,Z) & \equiv & \frac{1}{N} \sum_{i=1}^N \omega_i'  \left( -\log( n_{yz,i}') + \log(N-1) + \log c_{d_y+d_z} + (d_y+d_z) \log \rho_{k,i} \right). \\
 \hat{h}_N^q(Z) & \equiv & \frac{1}{N} \sum_{i=1}^N \omega_i' \left( -\log( n_{z,i}') + \log(N-1) + \log c_{d_z} + d_z \log \rho_{k,i} \right).
\end{eqnarray}

The term $ \sum_{i=1}^N \frac{\omega'_i}{N} \left( \log(N-1) + \psi(N) \right) $ converges to $0$ as $N \longrightarrow \infty$. We will prove the convergence of the rest of terms in the following lemmas, proving the Theorem \ref{theorem:ucmi_convergence}.
 
\begin{lemma} \label{lemma:hxyz_convergence}
Under the Assumption \ref{assumptions}, $\hat{h}_N^q(X,Y,Z) \overset{p}{\to} h^q(X,Y,Z)$ as $N \longrightarrow \infty$.
\end{lemma}
\begin{proof}
It directly follows from the Lemma 2 from \cite{gao2016conditional}. We just need to let $\tilde{X} \equiv (X,Z)$ and then show that $\hat{h}_N^q(\tilde{X},Y) \longrightarrow h^q(\tilde{X},Y)$ directly using the Lemma 2 from \cite{gao2016conditional}.  
\end{proof}

\begin{lemma} \label{lemma:3_terms_convergence}
For $N \longrightarrow \infty$,
\begin{equation} \hat{h}_N^q(X,Z) + \hat{h}_N^q(Y,Z) - \hat{h}_N^q(Z) \overset{p}{\to} h^q(X,Z) + h^q(Y,Z) - h^q(Z).\end{equation}
\end{lemma}

\section{Proof of lemma \ref{lemma:first_term_convergence} } \label{sec:first_term_convergence} 

Th proof is analog to that of Lemma 1 in \cite{gao2016conditional}. The term \eqref{eq:ucmi_first_line} is upper-bounded as,
\begin{eqnarray}
& & | \hat{I}^q(X;Y|Z) - \frac{1}{N} \sum_{i=1}^N \omega_i' g'(x_i,y_i,z_i) | \\
& = & | \frac{1}{N} \sum_{i=1}^N \left( \omega_i g(x_i,y_i,z_i) -\omega_i' g'(x_i,y_i,z_i) \right) | \\
& \leq & \frac{1}{N} \sum_{i=1}^N | \omega_i g(x_i,y_i,z_i) - \omega_i' g'(x_i,y_i,z_i) | \\
& \leq & \frac{1}{N} \sum_{i=1}^N \left( | \omega_i -  \omega'_i | |g'(x_i,y_i,z_i)| - \omega_i |g(x_i,y_i,z_i) - g'(x_i,y_i,z_i) | \right) \\
& = & \frac{1}{N} \sum_{i=1}^N \left( | \omega_i -  \omega'_i | |g'(x_i,y_i,z_i)| + \omega_i | \log(n_{yz,i}) - \log(n'_{yz,i}) | + \omega_i | \log(n_{z,i}) - \log(n'_{z,i}) | \right) \\
& \leq & \frac{1}{N} \sum_{i=1}^N \left( | \omega_i -  \omega'_i | |g'(x_i,y_i,z_i)| + \omega_i | n_{yz,i} - n'_{yz,i} |(\frac{1}{2n_{yz,i}}+\frac{1}{2n'_{yz,i}}) 
+ \omega_i | n_{z,i} - n'_{z,i} |(\frac{1}{2n_{z,i}}+\frac{1}{2n'_{z,i}}) \right) \\
& \leq & \frac{1}{N} \sum_{i=1}^N | \omega_i -  \omega'_i | |g'(x_i,y_i,z_i)| + \sum_{i=1}^N \frac{\omega_i}{N} \left( \frac{\max_{1\leq j\leq N}|\omega'_i-\omega_i|}{\min_{1\leq j\leq N}\omega'_i}+ \frac{\max_{1\leq j\leq N}|\omega'_i-\omega_i|}{\min_{1\leq j\leq N}\omega_i}  \right) \\
& \leq & \max_{1\leq i\leq N} | \omega_i -  \omega'_i | \left(  \max_{1\leq i\leq N} |g'(x_i,y_i,z_i)| + \frac{1}{\min_{1\leq j\leq N}\omega'_i}+ \frac{1}{\min_{1\leq j\leq N}\omega_i}  \right).
\end{eqnarray}

The term $g'(x_i,y_i,z_i)$ can be easily lower-bounded and upper-bounded as,
\begin{equation} -2 \log N \leq g'(x_i,y_i,z_i) \leq 2 \log N. \end{equation}

Thus, for any $\epsilon > 0$ and sufficiently large $N$ such that $\log N > \max \{ C_2 \epsilon / 3, 3 C_2\}$ and, if $| \omega_i -  \omega'_i | < \epsilon/(3 \log N)$ for all $i$ , we have,
\begin{eqnarray}
& & \max_{1\leq i\leq N} | \omega_i -  \omega'_i | \left(  \max_{1\leq i\leq N} |g'(x_i,y_i,z_i)| + \frac{1}{\min_{1\leq j\leq N}\omega'_i}+ \frac{1}{\min_{1\leq j\leq N}\omega_i}  \right) \\
& \leq & \frac{\epsilon}{3 \log N} \left( 2 \log N + C_2 + \frac{1}{1/C_2-\frac{\epsilon}{3 \log N}}  \right) \\
& \leq & \frac{\epsilon}{3 \log N} \left( 2 \log N + C_2 + 2C_2  \right) \leq \epsilon.
\end{eqnarray}

So for any $\epsilon > 0$ and sufficiently large $N$:
\begin{equation} P\left( \frac{1}{N} \sum_{i=1}^N | \omega_i g(x_i,y_i,z_i) - \omega_i' g'(x_i,y_i,z_i) | > \epsilon \right) \leq P \left( \max_{1\leq i\leq N} | \omega_i -  \omega'_i | > \frac{\epsilon}{3 \log N} \right). \end{equation}

Following the proof of Lemma 1 in \cite{gao2016conditional}, the term $P \left( \max_{1\leq i\leq N} | \omega_i -  \omega'_i | > \frac{\epsilon}{3 \log N} \right)$ converges to zero as $N \longrightarrow \infty$. So the desired convergence for the term \eqref{eq:ucmi_first_line} is obtained.  

\section{Proof of lemma \ref{lemma:3_terms_convergence} }

If we define,
\begin{eqnarray}
\hat{f}_{XZ}(x_i,z_i) & \equiv & \frac{n_{xz,i}}{(N-1)c_{d_x+d_z}\rho_{k,i}^{d_x+d_z}}, \\
\hat{f}^q_{YZ}(y_i,z_i) & \equiv & \frac{n'_{yz,i}}{(N-1)c_{d_y+d_z}\rho_{k,i}^{d_y+d_z}}, \\
\hat{f}^q_Z(z_i) & \equiv & \frac{n'_{z,i}}{(N-1)c_{d_z}\rho_{k,i}^{d_z}}, \\
\hat{a}_i & \equiv & \log \hat{f}_{XZ}(x_i,z_i) + \log \hat{f}_{YZ}^q(y_i,z_i) - \log \hat{f}_{Z}^q(z_i), \\
a_i & \equiv & \log f_{XZ}(x_i,z_i) + \log f_{YZ}^q(y_i,z_i) - \log f_{Z}^q(z_i).
\end{eqnarray}

Then,
\begin{eqnarray}
\hat{h}_N^q(X,Z) + \hat{h}_N^q(Y,Z) - \hat{h}_N^q(Z) & = &
 - \sum_{i=1}^N \frac{\omega'_i}{N} \left( \log \hat{f}_{XZ}(x_i,z_i) + \log \hat{f}_{YZ}^q(y_i,z_i) - \log \hat{f}_{Z}^q(z_i) \right) \\
& = & -\sum_{i=1}^N \frac{\omega'_i}{N} \hat{a}_i.
\end{eqnarray}

Now we can write,
\begin{eqnarray}
 & & | \hat{h}_N^q(X,Z) + \hat{h}_N^q(Y,Z) - \hat{h}_N^q(Z) - \left( h^q(X,Z) + h^q(Y,Z) - h^q(Z) \right)   | \\
 & \leq & 
| h^q(X,Z) + h^q(Y,Z) - h^q(Z) - \sum_{i=1}^N -\frac{\omega'_i}{N} a_i | \label{eq:lemma_3terms_1st_term} \\
&+& \sum_{i=1}^N \frac{\omega'_i}{N} | a_i  - \hat{a}_i |. \label{eq:lemma_3terms_2nd_term}
\end{eqnarray}

For the term \eqref{eq:lemma_3terms_1st_term}, since the terms $ a_i=\omega'_i \left( \log f_{XZ}(x_i,z_i) + \log f^q_{YZ}(y_i,z_i) + \log f^q_Z(z_i) \right)$ are i.i.d random variables, given the Assumption \ref{assumptions}, by the strong law of large numbers, we can write,
\begin{eqnarray}
\sum_{i=1}^N -\frac{\omega'_i}{N} a_i &=& \sum_{i=1}^N -\frac{\omega'_i}{N} \left( \log f_{XZ}(x,z) + \log f^q_{YZ}(y,z) - \log f^q_{Z}(z) \right) \\
& \longrightarrow & E \left[ -\frac{q_{XZ}(x,z)}{f_{XZ}(x,z)} \left( \log f_{XZ}(x,z) + \log f^q_{YZ}(y,z) - \log f^q_{Z}(z) \right) \right] \\
& = & -\int f_{Y|XZ}(y|x,z) q_{XZ}(x,z) \left( \log f_{XZ}(x,z) + \log f^q_{YZ}(y,z) - \log f^q_{Z}(z) \right) dx dy dz \\
& = &  h^q(X,Z) + h^q(Y,Z) - h^q(Z).
\end{eqnarray}   

Therefore, the term \eqref{eq:lemma_3terms_1st_term} converges to $0$ almost surely. For the term \eqref{eq:lemma_3terms_2nd_term}, let $T_i=(X_i,Y_i,Z_i)$, thus $t=(x,y,z)$, and $f_T(t)=f_{XYZ}(x,y,z)$. For any fixed $\epsilon>0$, we can write,
\begin{eqnarray} 
& & P \left( \sum_{i=1}^N \frac{\omega'_i}{N} | a_i  - \hat{a}_i | > \epsilon \right) \\
& \leq & P \left( \bigcup_{i=1}^N \{ | a_i  - \hat{a}_i | > \epsilon/2 \} \right) + P \left( \sum_{i=1}^N \omega'_i>2N \right).
\end{eqnarray}

The second term converges to zero. The first term can be bounded as,
\begin{eqnarray} 
& & P \left( \bigcup_{i=1}^N \{ | a_i  - \hat{a}_i | > \epsilon/2 \} \right) \\
& \leq & N.P \left( | a_i  - \hat{a}_i | > \epsilon/2  \right)  \leq N \int P \left( | a_i  - \hat{a}_i | > \epsilon/2 | T_i=t  \right) f_T(t) dt.
\end{eqnarray}

The term $P \left( | a_i  - \hat{a}_i | > \epsilon/2 | T_i=t  \right)$ can be upper-bounded by $I_1(t)+I_2(t)+I_3(t)+I_4(t)+I_5(t)$, where,
\begin{eqnarray} 
I_1(t) & = & P \left( \rho_{k,i} > r_1 | T_i=t \right). \\
I_2(t) & = & P \left( \rho_{k,i} < r_2 | T_i=t \right). \\
I_3(t) & = & \int_{r=r_2}^{r_1} P( | \log f_{XZ}(x_i,z_i) - \log \hat{f}_{XZ}(x_i,z_i) | > \epsilon/6 | \rho_{k,i}=r, T_i=t ) f_{\rho}(r) dr. \\
I_4(t) & = & \int_{r=r_2}^{r_1} P( | \log f_{U(YZ)}(y_i,z_i) - \log \hat{f}_{U(YZ)}(y_i,z_i) | > \epsilon/6 | \rho_{k,i}=r, T_i=t ) f_{\rho}(r) dr. \\
I_5(t) & = & \int_{r=r_2}^{r_1} P( | \log f_{U(Z)}(z_i) - \log \hat{f}_{U(Z)}(z_i) | > \epsilon/6 | \rho_{k,i}=r, T_i=t ) f_{\rho}(r) dr.
\end{eqnarray}

In which,
\begin{eqnarray}
r_1 & \equiv & \log N ( N f_T(t) c_{d_x+d_y+d_z} )^{ \frac{-1}{d_x+d_y+d_z} } \\
r_2 & \equiv & \max \{ (\log N)^2 ( N f_{XZ}(x,z) c_{d_x+d_z} )^{ \frac{-1}{d_x+d_z} }, (\log N)^2 ( N f^q(y,z) c_{d_y+d_z} )^{ \frac{-1}{d_y+d_z} }, (\log N)^2 ( N f_Z^q(z) c_{d_z} )^{ \frac{-1}{d_z} } \}.
\end{eqnarray}

The terms $I_1(t)$ and $I_2(t)$ represent the probability that the value of $\rho_{k.i}$ is too large or two small. We will show that both probabilities converge to $0$, i.e. $\rho_{i,k}$ obtained lies within a reasonable range. The $r_1$ threshold is determined based on the fact that $\frac{k}{N} \approx f_T(t_i) c_{d_x+d_y+d_z} \rho_{k,i}^{d_x+d_y+d_z} $ and selecting $k = \log N$ is a reasonable choice. The $r_2$ threshold, on the other hand, implies that $\rho_{i,k}$ should lie in a range such that $\frac{k_s}{N} \approx f_S(s_i) c_{d_s} \rho_{k,i}^{d_s}$ for all subspaces $s \in \{ (x,z),(y,z),z \}$.   

The terms $I_3(t)$, $I_4(t)$ and $I_5(t)$ represent the estimation error probability for a reasonable $\rho_{k,i}$.   

Considering each of the terms separately, We will prove that each term will go to zero as $N$ goes to infinity.

\subsubsection{Convergence of $I_1$}
The term $I_1(t)$ can be upper-bounded in the same way as explained in \cite{gao2016conditional} for $I_1(z)$. Then we will have,
\begin{equation}
I_1(t) \leq kN^{k-1} \exp \{ -\frac{(\log N)^{d_x+d_y+d_z}}{4} \}.
\end{equation}

\subsubsection{Convergence of $I_2$}

For the convergence of $I_2(t)$, we are goona take the same steps as \cite{gao2016conditional} for $I_2(z)$. First let $B_T(t,r) \equiv {n: \| n - t \| < r }$ be the $(d_x+d_y+d_z)$-dimensional ball centered at $z$ with radius $r$. For $r_{2,1} \equiv (\log N)^2 ( N f_{XZ}(x,z) c_{d_x+d_z} )^{ \frac{-1}{d_x+d_z} }$ and for sufficiently large $N$, the probability mass within the $B_T(t,r_{2,1})$ is given by,
\begin{eqnarray}
p_{2,1} & \equiv & P \left( u \in B_T(t, (\log N)^2 ( N f_{XZ}(x,z) c_{d_x+d_z} )^{ \frac{-1}{d_x+d_z} } \right)  \nonumber \\
& \leq & f_T(t) c_{d_x+d_y+d_z} \left( (\log N)^2 ( N f_{XZ}(x,z) c_{d_x+d_z} )^{ \frac{-1}{d_x+d_z} } \right)^{d_x+d_y+d_z} \left(1+ C (\log N)^4 ( N f_{XZ}(x,z) c_{d_x+d_z} )^{ \frac{-1}{d_x+d_z} })^2 \right) \nonumber \\
& \leq & \frac{2f_T(t) c_{d_x+d_y+d_z}}{(f_{XZ}(x,z) c_{d_x+d_z})^{\frac{d_x+d_y+d_z}{d_x+d_z}}} (\log N)^{2(d_x+d_y+d_z)} N^{-\frac{d_x+d_y+d_z}{d_x+d_z}} \\
& \leq & 2f_{Y|XZ}(y|x,z) \frac{c_{d_x+d_y+d_z}}{c_{d_x+d_z}} (\log N)^{2(d_x+d_y+d_z)} N^{-\frac{d_x+d_y+d_z}{d_x+d_z}} \\
& \leq & 2C' \frac{c_{d_x+d_y+d_z}}{c_{d_x+d_z}} (\log N)^{2(d_x+d_y+d_z)} N^{-\frac{d_x+d_y+d_z}{d_x+d_z}}.
\end{eqnarray}

Where the last inequalty comes from the assumption that $f_{Y|XZ}(y|x,z)<C'$ . Similarly, for the second threshold $r_{2,2}= (\log N)^2 ( N f^q(y,z) c_{d_y+d_z} )^{ \frac{-1}{d_y+d_z} }$,
\begin{eqnarray}
p_{2,2} & \equiv & P \left( u \in B_T(t, (\log N)^2 ( N f_{YZ}^q(y,z) c_{d_y+d_z} )^{ \frac{-1}{d_y+d_z} } \right)  \nonumber  \\
& \leq & f_T(t) c_{d_x+d_y+d_z} \left( (\log N)^2 ( N f_{YZ}^q(y,z) c_{d_y+d_z} )^{ \frac{-1}{d_y+d_z} } \right)^{d_x+d_y+d_z} \left(1+ C (\log N)^4 ( N f_{YZ}^q(y,z) c_{d_y+d_z} )^{ \frac{-1}{d_y+d_z} })^2 \right)   \nonumber \\
& \leq & \frac{2f_T(t) c_{d_x+d_y+d_z}}{(f_{YZ}^q(y,z) c_{d_y+d_z})^{\frac{d_x+d_y+d_z}{d_y+d_z}}} (\log N)^{2(d_x+d_y+d_z)} N^{-\frac{d_x+d_y+d_z}{d_y+d_z}} \\
& \leq & 2\frac{f_T(t)}{f_{YZ}^q(y,z)} \frac{c_{d_x+d_y+d_z}}{c_{d_y+d_z}} (\log N)^{2(d_x+d_y+d_z)} N^{-\frac{d_x+d_y+d_z}{d_y+d_z}} \\
& \leq & 2C_2\frac{f_T^q(t)}{f_{YZ}^q(y,z)} \frac{c_{d_x+d_y+d_z}}{c_{d_x+d_z}} (\log N)^{2(d_x+d_y+d_z)} N^{-\frac{d_x+d_y+d_z}{d_y+d_z}} \\
& \leq & 2C_2C' \frac{c_{d_x+d_y+d_z}}{c_{d_y+d_z}} (\log N)^{2(d_x+d_y+d_z)} N^{-\frac{d_x+d_y+d_z}{d_y+d_z}}.
\end{eqnarray}

The last two inequalities come from the bounds assumed on the distribution functions. Similarly, for the second threshold $r_{2,2}= (\log N)^2 ( N f_Z^q(z) c_{d_z} )^{ \frac{-1}{d_z}}$ we can write:
\begin{eqnarray}
p_{2,3} & \equiv & P \left( u \in B_T(t, (\log N)^2 ( N f_Z^q(z) c_{d_z} )^{ \frac{-1}{d_z} } \right) \\
& \leq & f_T(t) c_{d_x+d_y+d_z} \left( (\log N)^2 ( N f_Z^q(z) c_{d_z} )^{ \frac{-1}{d_z} } \right)^{d_x+d_y+d_z} \left(1+ C (\log N)^4 ( N f_Z^q(z) c_{d_z} )^{ \frac{-1}{d_z} })^2 \right) \\
& \leq & \frac{2f_T(t) c_{d_x+d_y+d_z}}{(f_Z^q(z) c_{d_z})^{\frac{d_x+d_y+d_z}{d_z}}} (\log N)^{2(d_x+d_y+d_z)} N^{-\frac{d_x+d_y+d_z}{d_z}} \\
& \leq & 2\frac{f_T(t)}{f_Z^q(z)} \frac{c_{d_x+d_y+d_z}}{c_{d_z}} (\log N)^{2(d_x+d_y+d_z)} N^{-\frac{d_x+d_y+d_z}{d_z}} \\
& \leq & 2C_2\frac{f_T^q(t)}{f_Z^q(z)} \frac{c_{d_x+d_y+d_z}}{c_{d_z}} (\log N)^{2(d_x+d_y+d_z)} N^{-\frac{d_x+d_y+d_z}{d_z}} \\
& \leq & 2C_2C' \frac{c_{d_x+d_y+d_z}}{c_{d_z}} (\log N)^{2(d_x+d_y+d_z)} N^{-\frac{d_x+d_y+d_z}{d_z}}
\end{eqnarray}
The last two inequalities come from the bounds assumed on the distribution functions.

Similar to the procedure for $I_2(z)$ in the \cite{gao2016conditional}, $I_2(t)$ is the probability that at least $k$ samples lie in $B_T(t,\max \{ r_{2,1}, r_{2,2}, r_{2,3} \} )$. Then we have,
\begin{eqnarray}
I_2(t) & = & P \left( \rho_{k,i} < \max \{ r_{2,1}, r_{2,2}, r_{2,3} \} \right) \\
& = & \sum_{m=k}^{N-1} \left( \begin{tabular}{c} N-1 \\ m \end{tabular} \right) \max \{ p_{2,1}, p_{2,2}, p_{2,3} \}^m \left( 1-\max \{ p_{2,1}, p_{2,2}, p_{2,3} \} \right)^{N-1-m} \\
& \leq & \sum_{m=k}^{N-1} N^m \max \{ p_{2,1}, p_{2,2}, p_{2,3} \}^m \\
& \leq & \sum_{m=k}^{N-1} \left( 2CC' \frac{c_{d_x+d_y+d_z}}{\min\{c_{d_y+d_z},c_{d_x+d_z},c_{d_z} \}} ( \log N )^{2(d_x+d_y+d_z)} N^{-\min\{ \frac{d_x+d_y}{d_z}, \frac{d_z}{d_x+d_y}, \frac{d_y}{d_x+d_z} \} } \right)^m. \\
\end{eqnarray}

Similarly, for $N$ sufficently large and applying the sum of geometric series, we have,
\begin{equation}
I_2(t) \leq \left( 4CC' \frac{c_{d_x+d_y+d_z}}{\min\{c_{d_y+d_z},c_{d_x+d_z},c_{d_z} \}}  \right)^k ( \log N )^{2k(d_x+d_y+d_z)} N^{-k\min\{ \frac{d_x+d_y}{d_z}, \frac{d_z}{d_x+d_y}, \frac{d_y}{d_x+d_z} \} }.
\end{equation}  

\subsubsection{Convergence of $I_3$}

Given $T_i = t = (x,y,z)$, and $\rho_{k,i}=r$ and $\hat{f}_{XZ}(x_i,z_i) = \frac{n_{xz,i}}{(N-1)c_{d_x+d_z}\rho_{k,i}^{d_x+d_z}}$, we have,
\begin{eqnarray}
& & P\left( | \log f_{XZ}(X_i,Z_i) - \log \hat{f}_{XZ}(X_i,Z_i) | > \epsilon/6 | \rho_{k,i}=r, T_i=t \right) \\
&=& P\left( n_{xz,i} > (N-1)c_{d_x+d_z}r^{d_x+d_z} f_{XZ}(x,z)e^{\epsilon/6} | \rho_{k,i}=r, T_i=t \right) \\
&+& P\left( n_{xz,i} < (N-1)c_{d_x+d_z}r^{d_x+d_z} f_{XZ}(x,z)e^{\epsilon/6} | \rho_{k,i}=r, T_i=t \right).
\end{eqnarray}

Given $T_i=t$, Lemma \ref{lemma:dist_nxz} gives the probability distribution of the $n_{xz,i}$. 

\begin{lemma} \label{lemma:dist_nxz}
Given $T_i = t = (x,y,z)$, and $\rho_{k,i}=r<r_N$ for some deterministic sequance of $r_N$ such that $\lim_{N \longleftarrow \infty} r_N = 0$ and for any $\epsilon > 0$, the number of neighbors $n_{xz,i}-k$ is distributed as $\sum_{l=k+1}^{N-1} U_l$, where $U_l$ are i.i.d Bernoulli random variables with mean $f_{XZ}(x,z)c_{d_x+d_z}r^{d_x+d_z} (1-\epsilon/8) \leq E [ U_l ] \leq f_{XZ}(x,z)c_{d_x+d_z}r^{d_x+d_z} (1-\epsilon/8) $ for sufficiently large $N$. 
\end{lemma}
\begin{proof}
See the proof of Lemma 5 in \cite{gao2016conditional}. 
\end{proof}

Based on Lemma \ref{lemma:dist_nxz},
\begin{eqnarray}
& & P\left( n_{xz,i} > (N-1)c_{d_x+d_z}r^{d_x+d_z} f_{XZ}(x,z)e^{\epsilon/6} | \rho_{k,i}=r, T_i=t \right) \\
&=& P\left( \sum_{l=k+1}^{N-1} U_l > (N-1)c_{d_x+d_z}r^{d_x+d_z} f_{XZ}(x,z)e^{\epsilon/6}-k \right) \\
&=& P\left( \sum_{l=k+1}^{N-1} U_l -(N-1-k)E[U_l] > (N-1)c_{d_x+d_z}r^{d_x+d_z} f_{XZ}(x,z)e^{\epsilon/6}-k-(N-1-k)E[U_l] \right).\label{eq:prob_I3}
\end{eqnarray}

The right hand side term inside the probability can be lower bounded as,
\begin{eqnarray}
& & (N-1)c_{d_x+d_z}r^{d_x+d_z} f_{XZ}(x,z)e^{\epsilon/6}-k-(N-1-k)E[U_l] \\
& \geq & (N-1)c_{d_x+d_z}r^{d_x+d_z} f_{XZ}(x,z)e^{\epsilon/6}-k-(N-1-k)f_{XZ}(x,z)c_{d_x+d_z}r^{d_x+d_z} (1-\epsilon/8) \\
& \geq & (N-k-1)c_{d_x+d_z}r^{d_x+d_z} f_{XZ}(x,z) \left( e^{\epsilon/6}-1-\epsilon/8 \right) -k \\
& \geq & (N-k-1)c_{d_x+d_z}r^{d_x+d_z}f_{XZ}(x,z)\frac{\epsilon}{48}.
\end{eqnarray}
 for sufficiently large $N$. 

Applying Bernstein's inequality, \eqref{eq:prob_I3} can be upper bounded by $\exp \{ -\frac{\epsilon^2}{2304(1+19\epsilon/144)}(N-k-1) c_{d_x+d_z} r^{d_x+d_z} f_{XZ}(x,z) \}$. The tail distribution can also be upper-bounded by the same term. Thus,
\begin{eqnarray}
 & & P( | \log f_{XZ}(x_i,z_i) - \log \hat{f}_{XZ}(x_i,z_i) | > \epsilon/6 | \rho_{k,i}=r, T_i=t ) \nonumber \\
&\leq& 2\exp \{ -\frac{\epsilon^2}{2304(1+19\epsilon/144)}(N-k-1) c_{d_x+d_z} r^{d_x+d_z} f_{XZ}(x,z) \}.
\end{eqnarray}

Therefore, the term $I_3(t)$ can be upper-bounded as,
\begin{eqnarray} 
I_3(t) & = & \int_{r=r_2}^{r_1} P( | \log f_{XZ}(x_i,z_i) - \log \hat{f}_{XZ}(x_i,z_i) | > \epsilon/6 | \rho_{k,i}=r, T_i=t ) f_{\rho}(r) dr \\
& \leq & \int_{r=(\log N)^2 ( N f_{XZ}(x,z) c_{d_x+d_z} )^{ \frac{-1}{d_x+d_z} }}^{\log N ( N f_T(t) c_{d_x+d_y+d_z} )^{ \frac{-1}{d_x+d_y+d_z} }} 
P( | \log f_{XZ}(x_i,z_i) - \log \hat{f}_{XZ}(x_i,z_i) | > \epsilon/6 | \rho_{k,i}=r, T_i=t ) f_{\rho}(r) dr  \nonumber \\
& \leq & \int_{r=(\log N)^2 ( N f_{XZ}(x,z) c_{d_x+d_z} )^{ \frac{-1}{d_x+d_z} }}^{\log N ( N f_T(t) c_{d_x+d_y+d_z} )^{ \frac{-1}{d_x+d_y+d_z} }} 
2\exp \{ -\frac{\epsilon^2}{2304(1+19\epsilon/144)}(N-k-1) c_{d_x+d_z} r^{d_x+d_z} f_{XZ}(x,z) \} f_{\rho}(r) dr  \nonumber  \\
& \leq & 2\exp \{ -\frac{\epsilon^2}{4608}N c_{d_x+d_z} f_{XZ}(x,z) \left( (\log N)^2 ( N f_{XZ}(x,z) c_{d_x+d_z} )^{ \frac{-1}{d_x+d_z} }\right)^{d_x+d_z} \} \\
& \leq & 2\exp \{ -\frac{\epsilon^2}{4608}\left( \log N \right)^{2(d_x+d_z)} \}.
\end{eqnarray}
For sufficiently large N.

\subsubsection{Convergence of $I_4$}

Given $T_i = t = (x,y,z)$, and $\rho_{k,i}=r$ and $\hat{f}_{U}(y_i,z_i) = \frac{n'_{yz,i}}{(N-1)c_{d_y+d_z}\rho_{k,i}^{d_y+d_z}}$, we have,
\begin{eqnarray}
& & P\left( | \log f^q(Y_i,Z_i) - \log \hat{f}^q(Y_i,Z_i) | > \epsilon/6 | \rho_{k,i}=r, T_i=t \right) \\
&=& P\left( n_{yz,i} > (N-1)c_{d_y+d_z}r^{d_y+d_z} f^q(Y_i,Z_i)e^{\epsilon/6} | \rho_{k,i}=r, T_i=t \right) \\
&+& P\left( n_{yz,i} < (N-1)c_{d_y+d_z}r^{d_y+d_z} f^q(Y_i,Z_i)e^{\epsilon/6} | \rho_{k,i}=r, T_i=t \right).
\end{eqnarray}

We can write $n'_{yz,i}=n_{yz,i}^{(1)}+n_{yz,i}^{(2)}$, where,
\begin{eqnarray}
n_{yz,i}^{(1)} &=& \sum_{j:\| T_j-t \| < \rho_{i,k} } \frac{q_{XZ}(x_j,z_j)}{f_{XZ}(x_j,z_j)}  \\
n_{yz,i}^{(2)} &=& \sum_{j:\| T_j-t \| > \rho_{i,k} } \frac{q_{XZ}(x_j,z_j)}{f_{XZ}(x_j,z_j)} I \{ \| (Y_j-Y_i,Z_j-Z_i) \| < \rho_{k,i} \}.
\end{eqnarray}

Given $T_i=t$, Lemma \ref{lemma:dist_nyz} gives the probability distribution of the $n'_{yz,i}$. 

\begin{lemma} \label{lemma:dist_nyz}
Given $T_i = t = (x,y,z)$, and $\rho_{k,i}=r<r_N$ for some deterministic sequance of $r_N$ such that $\lim_{N \longleftarrow \infty} r_N = 0$ and for any $\epsilon > 0$, the distribution of $n_{yz,i}^{(2)}$ is $\sum_{l=k+1}^{N-1} V_l$, where $V_l$ are i.i.d random variables with $V_l \in [0,1/C_1]$ and mean $f_{XZ}(x,z)c_{d_x+d_z}r^{d_x+d_z} (1-\epsilon/8) \leq E [ U_l ] \leq f_{XZ}(x,z)c_{d_x+d_z}r^{d_x+d_z} (1-\epsilon/8) $ for sufficiently large $N$. 
\end{lemma}
\begin{proof}
See the proof of Lemma 6 in \cite{gao2016conditional}. 
\end{proof}

According to the Lemma \ref{lemma:dist_nyz} and following the same procedure as $I_3(t)$, the term $I_4(t)$ will also be bounded here in th same way. We have:
\begin{equation}I_4(t) \leq 2\exp \{ -\frac{C_1 \epsilon^2}{4608}\left( \log N \right)^{2(d_y+d_z)} \} \end{equation}

\subsubsection{Convergence of $I_5$}

Similar to the case of $I_4(t)$, given $T_i = t = (x,y,z)$, and $\rho_{k,i}=r$ and $\hat{f}_{U}(z_i) = \frac{n'_{z,i}}{(N-1)c_{d_z}\rho_{k,i}^{d_z}}$, we have,
\begin{eqnarray}
& & P\left( | \log f^q(Z_i) - \log \hat{f}^q(Z_i) | > \epsilon/6 | \rho_{k,i}=r, T_i=t \right) \\
&=& P\left( n_{z,i} > (N-1)c_{d_z}r^{d_z} f^q(Z_i)e^{\epsilon/6} | \rho_{k,i}=r, T_i=t \right) \\
&+& P\left( n_{z,i} < (N-1)c_{d_z}r^{d_z} f^q(Z_i)e^{\epsilon/6} | \rho_{k,i}=r, T_i=t \right).
\end{eqnarray}

We can write $n'_{z,i}=n_{z,i}^{(1)}+n_{z,i}^{(2)}$, where,
\begin{eqnarray}
n_{z,i}^{(1)} &=& \sum_{j:\| T_j-t \| < \rho_{i,k} } \frac{q_{XZ}(x_j,z_j)}{f_{XZ}(x_j,z_j)}  \\
n_{z,i}^{(2)} &=& \sum_{j:\| T_j-t \| > \rho_{i,k} } \frac{q_{XZ}(x_j,z_j)}{f_{XZ}(x_j,z_j)} I \{ \| Z_j-Z_i \| < \rho_{k,i} \}.
\end{eqnarray}

Following the same procedure as for $I_4(t)$, we will obtain the upper bound below for $I_5(t)$:
\begin{equation} I_5(t) \leq 2\exp \{ -\frac{C_1 \epsilon^2}{4608}\left( \log N \right)^{2(d_z)} \} \end{equation}

Now, we can write,
\begin{eqnarray} 
& & P \left( \sum_{i=1}^N \frac{\omega'_i}{N} | a_i  - \hat{a}_i | > \epsilon \right) \\
& \leq & N \int \left( I_1(t)+I_2(t)+I_3(t)+I_4(t)+I_5(t)  \right) f_T(t) dt \\
& \leq & kN^k \exp \{ -\frac{(\log N)^{d_x+d_y+d_z}}{4} \} \nonumber \\
& + & \left( 4CC' \frac{c_{d_x+d_y+d_z}}{\min\{c_{d_y+d_z},c_{d_x+d_z},c_{d_z} \}}  \right)^k ( \log N )^{2k(d_x+d_y+d_z)} N^{1-k\min\{ \frac{d_x+d_y}{d_z}, \frac{d_z}{d_x+d_y}, \frac{d_y}{d_x+d_z} \} } \nonumber \\
& + & 2N\exp \{ -\frac{\epsilon^2}{4608}\left( \log N \right)^{2(d_x+d_z)} + 4 N \exp \{ -\frac{C_1 \epsilon^2}{4608}\left( \log N \right)^{2(d_z)} \}.
\end{eqnarray}

If $k$ is chosen large enough so that $1-k\min\{ \frac{d_x+d_y}{d_z}, \frac{d_z}{d_x+d_y}, \frac{d_y}{d_x+d_z} \} <0$, then all the terms will converge to $0$ as $N$ goes to infinity, and the proof of Lemma \ref{lemma:3_terms_convergence} is complete.   




\bibliography{ref}

\begin{thebibliography}{10}

\bibitem{cover2012elements}
T.~M. Cover and J.~A. Thomas, {\em Elements of information theory}.
\newblock John Wiley \& Sons, 2012.

\bibitem{janzing2013quantifying}
D.~Janzing, D.~Balduzzi, M.~Grosse-Wentrup, and B.~Sch{\"o}lkopf, ``Quantifying
  causal influences,'' {\em The Annals of Statistics}, pp.~2324--2358, 2013.

\bibitem{spirtes2000causation}
P.~Spirtes, C.~N. Glymour, and R.~Scheines, {\em Causation, prediction, and
  search}.
\newblock MIT press, 2000.

\bibitem{holland1985statistics}
P.~W. Holland, C.~Glymour, and C.~Granger, ``Statistics and causal inference,''
  {\em ETS Research Report Series}, vol.~1985, no.~2, 1985.

\bibitem{dawid1979conditional}
A.~P. Dawid, ``Conditional independence in statistical theory,'' {\em Journal
  of the Royal Statistical Society. Series B (Methodological)}, pp.~1--31,
  1979.

\bibitem{granger1969investigating}
C.~W. Granger, ``Investigating causal relations by econometric models and
  cross-spectral methods,'' {\em Econometrica: Journal of the Econometric
  Society}, pp.~424--438, 1969.

\bibitem{geiger2015causal}
P.~Geiger, K.~Zhang, B.~Schoelkopf, M.~Gong, and D.~Janzing, ``Causal inference
  by identification of vector autoregressive processes with hidden
  components,'' in {\em International Conference on Machine Learning},
  pp.~1917--1925, 2015.

\bibitem{hosseini2016learning}
H.~Hosseini, S.~Kannan, B.~Zhang, and R.~Poovendran, ``Learning temporal
  dependence from time-series data with latent variables,'' in {\em Data
  Science and Advanced Analytics (DSAA), 2016 IEEE International Conference
  on}, pp.~253--262, IEEE, 2016.

\bibitem{eichler2012graphical}
M.~Eichler, ``Graphical modelling of multivariate time series,'' {\em
  Probability Theory and Related Fields}, vol.~153, no.~1-2, pp.~233--268,
  2012.

\bibitem{quinn2015directed}
C.~J. Quinn, N.~Kiyavash, and T.~P. Coleman, ``Directed information graphs,''
  {\em IEEE Transactions on information theory}, vol.~61, no.~12,
  pp.~6887--6909, 2015.

\bibitem{sun2015causal}
J.~Sun, D.~Taylor, and E.~M. Bollt, ``Causal network inference by optimal
  causation entropy,'' {\em SIAM Journal on Applied Dynamical Systems},
  vol.~14, no.~1, pp.~73--106, 2015.

\bibitem{rahimzamani2016network}
A.~Rahimzamani and S.~Kannan, ``Network inference using directed information:
  The deterministic limit,'' in {\em Communication, Control, and Computing
  (Allerton), 2016 54th Annual Allerton Conference on}, pp.~156--163, IEEE,
  2016.

\bibitem{qiu2012understanding}
X.~Qiu, S.~Ding, and T.~Shi, ``From understanding the development landscape of
  the canonical fate-switch pair to constructing a dynamic landscape for
  two-step neural differentiation,'' {\em PloS one}, vol.~7, no.~12, p.~e49271,
  2012.

\bibitem{gao2016causal}
W.~Gao, S.~Kannan, S.~Oh, and P.~Viswanath, ``Causal strength via shannon
  capacity: Axioms, estimators and applications,'' in {\em Proceedings of the
  33rd International Conference on Machine Learning}, 2016.

\bibitem{kim2017discovering}
H.~Kim, W.~Gao, S.~Kanan, S.~Oh, and P.~Viswanath, ``Discovering potential
  correlations via hypercontractivity,'' {\em arXiv preprint arXiv:1709.04024},
  2017.

\bibitem{anantharam2013maximal}
V.~Anantharam, A.~Gohari, S.~Kamath, and C.~Nair, ``On maximal correlation,
  hypercontractivity, and the data processing inequality studied by erkip and
  cover,'' {\em arXiv preprint arXiv:1304.6133}, 2013.

\bibitem{polyanskiy2017strong}
Y.~Polyanskiy and Y.~Wu, ``Strong data-processing inequalities for channels and
  bayesian networks,'' in {\em Convexity and Concentration}, pp.~211--249,
  Springer, 2017.

\bibitem{tishby2000information}
N.~Tishby, F.~C. Pereira, and W.~Bialek, ``The information bottleneck method,''
  {\em arXiv preprint physics/0004057}, 2000.

\bibitem{gao2016conditional}
W.~Gao, S.~Kannan, S.~Oh, and P.~Viswanath, ``Conditional dependence via
  shannon capacity: Axioms, estimators and applications,'' {\em arXiv preprint
  arXiv:1602.03476}, 2016.

\bibitem{kidambi2015shannon}
R.~Kidambi and S.~Kannan, ``On shannon capacity and causal estimation,'' in
  {\em Communication, Control, and Computing (Allerton), 2015 53rd Annual
  Allerton Conference on}, pp.~988--992, IEEE, 2015.

\bibitem{mooij2016distinguishing}
J.~M. Mooij, J.~Peters, D.~Janzing, J.~Zscheischler, and B.~Sch{\"o}lkopf,
  ``Distinguishing cause from effect using observational data: methods and
  benchmarks,'' {\em The Journal of Machine Learning Research}, vol.~17, no.~1,
  pp.~1103--1204, 2016.

\bibitem{devroye1984consistency}
L.~Devroye and C.~S. Penrod, ``The consistency of automatic kernel density
  estimates,'' {\em The Annals of Statistics}, pp.~1231--1249, 1984.

\bibitem{sheather1991reliable}
S.~J. Sheather and M.~C. Jones, ``A reliable data-based bandwidth selection
  method for kernel density estimation,'' {\em Journal of the Royal Statistical
  Society. Series B (Methodological)}, pp.~683--690, 1991.

\bibitem{Kra04}
A.~Kraskov, H.~St{\"o}gbauer, and P.~Grassberger, ``Estimating mutual
  information,'' {\em Physical review E}, vol.~69, no.~6, p.~066138, 2004.

\bibitem{khan2007relative}
S.~Khan, S.~Bandyopadhyay, A.~R. Ganguly, S.~Saigal, D.~J. Erickson~III,
  V.~Protopopescu, and G.~Ostrouchov, ``Relative performance of mutual
  information estimation methods for quantifying the dependence among short and
  noisy data,'' {\em Physical Review E}, vol.~76, no.~2, p.~026209, 2007.

\bibitem{kozachenko1987sample}
L.~Kozachenko and N.~N. Leonenko, ``Sample estimate of the entropy of a random
  vector,'' {\em Problemy Peredachi Informatsii}, vol.~23, no.~2, pp.~9--16,
  1987.

\bibitem{gao2017demystifying}
W.~Gao, S.~Oh, and P.~Viswanath, ``Demystifying fixed k-nearest neighbor
  information estimators,'' in {\em Information Theory (ISIT), 2017 IEEE
  International Symposium on}, pp.~1267--1271, IEEE, 2017.

\bibitem{gao2017estimating}
W.~Gao, S.~Kannan, S.~Oh, and P.~Viswanath, ``Estimating mutual information for
  discrete-continuous mixtures,'' {\em arXiv preprint arXiv:1709.06212}, 2017.

\bibitem{frenzel2007partial}
S.~Frenzel and B.~Pompe, ``Partial mutual information for coupling analysis of
  multivariate time series,'' {\em Physical review letters}, vol.~99, no.~20,
  p.~204101, 2007.

\bibitem{valiant2011estimating}
G.~Valiant and P.~Valiant, ``Estimating the unseen: an n/log (n)-sample
  estimator for entropy and support size, shown optimal via new clts,'' in {\em
  Proceedings of the forty-third annual ACM symposium on Theory of computing},
  pp.~685--694, ACM, 2011.

\bibitem{jiao2015minimax}
J.~Jiao, K.~Venkat, Y.~Han, and T.~Weissman, ``Minimax estimation of
  functionals of discrete distributions,'' {\em IEEE Transactions on
  Information Theory}, vol.~61, no.~5, pp.~2835--2885, 2015.

\bibitem{wu2016minimax}
Y.~Wu and P.~Yang, ``Minimax rates of entropy estimation on large alphabets via
  best polynomial approximation,'' {\em IEEE Transactions on Information
  Theory}, vol.~62, no.~6, pp.~3702--3720, 2016.

\bibitem{watanabe1960information}
S.~Watanabe, ``Information theoretical analysis of multivariate correlation,''
  {\em IBM Journal of research and development}, vol.~4, no.~1, pp.~66--82,
  1960.

\bibitem{chan2015multivariate}
C.~Chan, A.~Al-Bashabsheh, J.~B. Ebrahimi, T.~Kaced, and T.~Liu, ``Multivariate
  mutual information inspired by secret-key agreement,'' {\em Proceedings of
  the IEEE}, vol.~103, no.~10, pp.~1883--1913, 2015.

\bibitem{moon2017ensemble}
K.~R. Moon, K.~Sricharan, and A.~O. Hero~III, ``Ensemble estimation of mutual
  information,'' {\em arXiv preprint arXiv:1701.08083}, 2017.

\end{thebibliography}
\bibliographystyle{ieeetr}

\end{document}